\documentclass[paper]{ieice}
\usepackage[dvips]{graphicx}
\usepackage{theorem}
\usepackage{amsmath}
\usepackage{amssymb}
\usepackage{graphicx}
\usepackage{subfigure}
\theorembodyfont{\rmfamily}   
\newtheorem{theorem}{Theorem}
\newtheorem{lemma}{Lemma}
\newtheorem{definition}{Definition}
\newtheorem{corollary}{Corollary}

\newcommand{\qed}{\hfill$\square$}

\newenvironment{proof}{%
  \noindent{\em Proof.\ }}{%
  \hspace*{\fill}\qed \\
  \vspace{2ex}}

\field{A}
\title{Secret Key Agreement by Soft-decision
of Signals in Gaussian Maurer's Model}
\titlenote{A part of this paper will be presented at 2008 IEEE International
Symposium on Information Theory in Toronto, Canada.}
\authorlist{
 \authorentry[m-naito@sp.m.is.nagoya-u.ac.jp]{Masashi Naito}{n}{labelA}
 \authorentry[shun-wata@it.ss.titech.ac.jp]{Shun Watanabe}{m}{labelB}
 \authorentry[ryutaroh@rmatsumoto.org]{Ryutaroh Matsumoto}{m}{labelB}
 \authorentry[uyematsu@ieee.org]{Tomohiko Uyematsu}{m}{labelB}
}
\affiliate[labelA]{The author is with the Department of Media Science
Graduate School of Information Science, Nagoya University}
\affiliate[labelB]{The authors are with the Department of Communications
and Integrated Sysmtems, Tokyo Institute of Technology }

\received{2008}{4}{1}
\revised{2003}{1}{1}
\finalreceived{2003}{1}{1}

\begin{document}
\maketitle
\begin{summary}
We consider the problem of secret key agreement in Gaussian Maurer's
 Model.  In Gaussian Maurer's model, legitimate receivers, Alice and
 Bob, and a wire-tapper, Eve, receive signals randomly generated by a
 satellite through three independent memoryless Gaussian channels
 respectively.  Then Alice and Bob generate a common secret key from
 their received signals.  In this model, we propose a protocol for
 generating a common secret key by using the result of soft-decision of
 Alice and Bob's received signals.  Then, we calculate a lower bound on
 the secret key rate in our proposed protocol. As a result of comparison
 with the protocol that only uses hard-decision, we found that the
 higher rate is obtained by using our protocol.
\end{summary}

\section{Introduction}

As one of fundamental problems in cryptography, we will consider the
problem of secret key agreement in this paper.  That is to say, we will
consider how to generate a common secret key by two parties not sharing
such a key initially in the situation that a wire-tapper has access to
the communication channel between two parties.  Many models of this
problem were presented and and analyzed in the literatures
\cite{wyner:75,cheong:76,cheong:78,csiszar:78}.
Recently, key agreement over wireless channel is experimentally
studied \cite{aono:05}. 

Maurer \cite{maurer:93} and Ahlswede and Csis\'{z}ar \cite{ahlswede:93}
considered the interactive model of secret key agreement from an
initially shared partially secret string by communication over a public
channel.

Maurer \cite{maurer:93} considered the following model.  Two parties,
Alice and Bob, who want to share a secret key, and the wire-tapper, Eve,
receive the bits randomly generated by a satellite over independent
binary symmetric channels (BSC) respectively.  We call this model
Maurer's model.  Maurer \cite{maurer:93} proposed an interactive
protocol in his model, and he showed a lower bound on key rates at which
Alice and Bob can agree a secret key.  Note that the key rate is defined
as length of the secret key generated by Alice and Bob per channel use
by the satellite.

In Maurer's original model and protocol, channels are assumed to be BSC,
and received signals are assumed to be digital signals.  However,
signals in practical channels are analogue.  
Recently, key agreement over wireless channel is experimentally
studied by Aono et al.~\cite{aono:05}. However, information theoretic analysis
of the key agreement over analogue channels has not sufficiently conducted.
In order to close the gap between Maurer's results and the experimental study,
we will modify Maurer's model to use Gaussian channels
instead of BSC, which we call Gaussian Maurer's model.  

In Gaussian
Maurer's model, Alice and Bob can use the results of soft-decision of
analogue received signals.  They can determine the reliability
information from this results and use it for generating a common secret
key.
In this paper, we will propose a protocol for secret key agreement using
the reliability information.  Then, we calculate key rates at which
Alice and Bob can agree a secret key in our proposed protocol.

Considering the situation that Alice, Bob, and Eve hard-detect the
signals that are sent out by the satellite, Maurer's original model can
be seen as the special case of Gaussian Maurer's model.  Thus, we can
compare the protocol in Gaussian Maurer's model and one in BSC Maurer's
model.  In order to show advantage to use reliability information, we
will compare the key rate in our proposed protocol and the key rate in
Maurer's protocol in which Alice and Bob use only hard-decision.  that
uses only hard-decision.  From the result of this comparison, we will
show that the higher key rate is obtained by using our proposed protocol
than the protocol that only uses hard-decision.

Rest of this paper is organized as follows.  In section 2, we will
introduce Maurer's model modified to use Gaussian channels instead of
BSC.  In section 3, we will show our proposed protocol using reliability
information.  In section 4, we will compare our proposed protocol and
Maurer's protocol with hard-decision.  In appendices, we will prove the
lemmas that is needed for the proof of theorem that derives a lower
bound on key rates at which Alice and Bob can agree a secret key.

\section{Secret Key Rate in Gaussian Maurer's Model}

Consider the following key agreement problem, which we call Gaussian
Maurer's model.  Assume that a satellite randomly generates signals and
sends it to two parties Alice and Bob who want to share secret key and
the wire-tapper Eve over three independent memoryless Gaussian channels.
Their noises at time $i$, denoted $N_{A}^{(i)}$, $N_{B}^{(i)}$, and
$N_{E}^{(i)}$, are drawn from independently identically distributed
(i.i.d.) Gaussian distributions with mean $0$ and variances $V_{A}$,
$V_{B}$, and $V_{E}$ respectively.  A sequence of signals that the
satellite generates at time $1$ to $n$, denoted
$U^{n}=[U^{(1)},\dots,U^{(n)}]$, is drawn from a distribution
$P_{U^{n}}$ on a signal set in $\mathbb{R}^{n}$ and this sequence of
signals satisfies power constraint
$\frac{1}{n}\sum_{i=1}^{n}(u^{(i)})^{2} \leq 1$ for all sequences
$u^{n}$.  Alice, Bob, and Eve receive $X^{n}=[X^{(1)},\dots,X^{(n)}]$,
$Y^{n}=[Y^{(1)},\dots,Y^{(n)}]$,and $Z^{n}=[Z^{(1)},\dots,Z^{(n)}]$, as
outputs of these three channels at time $1$ to $n$ respectively.  They
are assumed to know the distribution $P_{U^{n}}$ and noise variances
$V_{A}$, $V_{B}$, and $V_{E}$.  Note that capital letters denote random
variables and corresponding small letters denote realizations in this
paper.

After Alice, Bob, and Eve receive signals, Alice and Bob communicate
over a public channel.  This channel is assumed to be noiseless and
discrete, and its capacity is finite.  Every messages communicated
between Alice and Bob can be intercepted by Eve, but it is assumed that
Eve cannot fraudulent messages nor modify messages on this public
channel without being detected.  Let $C$ be the entire communication
held over this public channel.  After enough communication over the
public channel, Alice computes a secret key $S$ on a key alphabet
$\mathcal{S}$ as a function of her received signals $X^{n}$ and all
information $C$ over the public channel.  In a similar way, Bob computes
a secret key $S'$ on $\mathcal{S}$ as a function of $Y^{n}$ and $C$.
The secret key rate in this model is defined as follows.  Note that we
will take all logarithms to be base 2, and hence all the entropies will
be measured in bits.

\begin{definition}
For given noise variances $V_{A}$, $V_{B}$, and $V_{E}$, a
rate $R$ is said to be \textit{achievable} if for every $\epsilon>0$
there exists a protocol for sufficiently large $n$ satisfying
\begin{gather}
\Pr[S \neq S']\leq \epsilon , \label{eq:agreement}\\ H(S|CZ^{n})\geq
\log|\mathcal{S}|-\epsilon \label{eq:security}
\end{gather}
and
\begin{equation}  
\frac{1}{n}\log|\mathcal{S}|\geq R-\epsilon,
\end{equation}
where $|\mathcal{S}|$ denotes the number of the elements in
$\mathcal{S}$.
\end{definition}

\begin{definition}
The \textit{secret key rate} for given noise variances
$V_{A}$, $V_{B}$, and $V_{E}$, denoted $R_{S}(V_{A},V_{B},V_{E})$, is
the supremum of all achievable rate.
\end{definition}

\section{Secret Key Agreement by Soft-Decision of Signals}

In this section, we will propose a protocol that uses reliability
information of signals and calculate a lower bound on the secret key
rate in this protocol.

In our proposed protocol, the satellite selects input signal $U^{(i)}$
i.i.d. according to a distribution $P_{U}(1)=P_{U}(-1)=\frac{1}{2}$.
Thus, the received signals $X^{(i)},Y^{(i)},Z^{(i)}$ are also
i.i.d. respectively.

Let $a_{1},\dots,a_{K}$ be a positive monotonically increasing sequence,
and let $E_{1},\dots,E_{K}$ be sets, where $j$th level set is defined as
$E_{j}=[-a_{j},a_{j}]$~ $(j=1,\dots,K)$.

The procedures of our proposed protocol is as follows.

\begin{enumerate}
 \item From the received signal $X^{(i)}$ at time $i$, Alice determines
reliability information $W_{A}^{(i)}$ as
\begin{equation}
W_{A}^{(i)}=\begin{cases} 0 & \text{if $X^{(i)} \in E_{1}$}\\ j &
               \text{if $X^{(i)} \in E_{j}^{c} \backslash E_{j+1}^{c}$
               ($j=1,\dots,K$)}\\ K & \text{if $X^{(i)} \in E_{K}^{c}$}
               \notag
            \end{cases} ,
\end{equation}   
where the set $E_{j}^{c}$ is the complementary set of the set $E_{j}$ in
the set of real numbers $\mathbb{R}$, and $E_{j}^{c} \backslash
E_{j+1}^{c} = E_{j}^{c} \cap E_{j+1}$ is the difference set.  Similarly,
from the received signal $Y^{(i)}$ at time $i$, Bob determine
reliability information $W_{B}^{(i)}$ as
\begin{equation}
W_{B}^{(i)}=\begin{cases} 0 & \text{if $Y^{(i)} \in E_{1}$}\\ j &
               \text{if $Y^{(i)} \in E_{j}^{c} \backslash E_{j+1}^{c}$
               ($j=1,\dots,K$)}\\ K & \text{if $Y^{(i)} \in E_{K}^{c}$}
               \notag
            \end{cases} .
\end{equation}   

 \item Alice and Bob send sequences
       $W_{A}^{n}=[W_{A}^{(1)},\dots,W_{A}^{(n)}]$ and
       $W_{B}^{n}=[W_{B}^{(1)},\dots,W_{B}^{(n)}]$ over the public
       channel.  From these messages, they can know the sets containing
       their received signals.

 \item Alice and Bob quantize $X^{n}$ and $Y^{n}$ into discrete random
variables $\Tilde{X}_{\Delta}^{n}$ and $\Tilde{Y}_{\Delta}^{n}$, where
$\Tilde{X}_{\Delta}^{(i)}$ is defined as
\begin{equation}
\Tilde{X}_{\Delta}^{(i)}=\begin{cases} 1 & \text{if $X^{(i)} \geq 0$,}\\
                           0 & \text{if $X^{(i)} <
                           0$,}\label{hard-decision1}
\end{cases}
\end{equation}
and $\Tilde{Y}_{\Delta}^{(i)}$ is similarly defined as
\begin{equation}
\Tilde{Y}_{\Delta}^{(i)}=\begin{cases} 1 & \text{if $Y^{(i)} \geq 0$,}\\
                           0 & \text{if $Y^{(i)} <
                           0$.}\label{hard-decision2}
\end{cases}
\end{equation}
\end{enumerate}
 
For given $(W_{A}^{(i)},W_{B}^{(i)})=(w_{A},w_{B})$, if Eve's ambiguity
$H(\Tilde{X}_{\Delta}|Z,W_{A}=w_{A},W_{B}=w_{B})$ about
$\Tilde{X}_{\Delta}^{(i)}$ is smaller than Bob's ambiguity
$H(\Tilde{X}_{\Delta}|Y,W_{A}=w_{A},W_{B}=w_{B})$ about
$\Tilde{X}_{\Delta}^{(i)}$, then we should discard
$\Tilde{X}_{\Delta}^{(i)}$ in our protocol.  Indeed, if we keep
$\Tilde{X}_{\Delta}^{(i)}$ for such
$(W_{A}^{(i)},W_{B}^{(i)})=(w_{A},w_{B})$, then a negative term is added
to the lower bound on a secret key rate shown in Eq.~(\ref{eq:rate}).
Furthermore, if the difference between Eve and Bob's ambiguity about
$\Tilde{X}_{\Delta}^{(i)}$ is smaller than the difference between Eve's
ambiguity $H(\Tilde{Y}_{\Delta}|Z,W_{A}=w_{A},W_{B}=w_{B})$ about
$\Tilde{Y}_{\Delta}^{(i)}$ and Alice's ambiguity
$H(\Tilde{Y}_{\Delta}|X,W_{A}=w_{A},W_{B}=w_{B})$ about
$\Tilde{Y}_{\Delta}^{(i)}$, we should generate a secret key from
$\Tilde{Y}_{\Delta}^{(i)}$ instead of $\Tilde{X}_{\Delta}^{(i)}$.  For
this purpose, we consider the sets $A,B \subset \{ 1,\dots,K \} \times
\{ 1,\dots,K \}$, which are defined as
\begin{align}
A&=\{(w_{A},w_{B})|\notag\\
 &\mspace{40mu}H(\Tilde{X}_{\Delta}|Z,W_{A}=w_{A},W_{B}=w_{B})\notag\\
 &\mspace{50mu}-H(\Tilde{X}_{\Delta}|Y,W_{A}=w_{A},W_{B}=w_{B})\notag\\
 &\mspace{70mu}\ge\max\{0,H(\Tilde{Y}_{\Delta}|Z,W_{A}=w_{A},W_{B}=w_{B})\notag\\
 &\mspace{130mu}-H(\Tilde{Y}_{\Delta}|X,W_{A}=w_{A},W_{B}=w_{B})\}\},\notag\\
 B&=\{(w_{A},w_{B})|\notag\\
 &\mspace{40mu}H(\Tilde{Y}_{\Delta}|Z,W_{A}=w_{A},W_{B}=w_{B})\notag\\
 &\mspace{50mu}-H(\Tilde{Y}_{\Delta}|X,W_{A}=w_{A},W_{B}=w_{B})\notag\\
 &\mspace{70mu}>\max\{0,H(\Tilde{X}_{\Delta}|Z,W_{A}=w_{A},W_{B}=w_{B})\notag\\
 &\mspace{130mu}-H(\Tilde{X}_{\Delta}|Y,W_{A}=w_{A},W_{B}=w_{B})\}\}.\notag
\end{align}
If given $(W_{A}^{(i)},W_{B}^{(i)})$ is in the set $A$, we use
$\Tilde{X}_{\Delta}^{(i)}$ for generating a secret key, otherwise we
discard $\Tilde{X}_{\Delta}^{(i)}$.  Similarly, if given
$(W_{A}^{(i)},W_{B}^{(i)})$ is in the set $B$, we use
$\Tilde{Y}_{\Delta}^{(i)}$ for generating a secret key, otherwise we
discard $\Tilde{Y}_{\Delta}^{(i)}$.  Thus, we determine discrete random
variables
\begin{equation}
X_{\Delta}^{(i)}=
\begin{cases}
  \Tilde{X}_{\Delta}^{(i)} & \text{if $(W_{A}^{(i)},W_{B}^{(i)})\in
  A$,}\\ 0 & \text{otherwise,}\label{detemination:x}
\end{cases}
\end{equation}
and
\begin{equation}
Y_{\Delta}^{(i)}=
\begin{cases}
  \Tilde{Y}_{\Delta}^{(i)} & \text{if $(W_{A}^{(i)},W_{B}^{(i)})\in
  B$,}\\ 0 & \text{otherwise,}\label{detemination:y}
\end{cases}
\end{equation}
and we use them for generating a secret key instead of
$\Tilde{X}_{\Delta}^{(i)}$ and $\Tilde{Y}_{\Delta}^{(i)}$.
\begin{enumerate}
\setcounter{enumi}{3}
 \item According to the rule in Eq.~(\ref{detemination:x}), Alice
       determines $X_{\Delta}^{n}$ from $W_{A}^{n}$, $W_{B}^{n}$, and
       $\Tilde{X}_{\Delta}^{n}$.  Similarly, Bob determines
       $Y_{\Delta}^{n}$ from $W_{A}^{n}$, $W_{B}^{n}$, and
       $\Tilde{Y}_{\Delta}^{n}$.
 \item Alice sends partial information of $X_{\Delta}^{n}$ as a public
       message $M_{A}$ on $\mathcal{M}_{A}$ in order to share
       $X_{\Delta}^{n}$ with Bob.  Similarly, Bob sends partial
       information of $Y_{\Delta}^{n}$ as a public message $M_{B}$ on
       $\mathcal{M}_{B}$.
 \item Alice decodes $M_{B}$, $X^{n}$, and the reliability information
       $(W_{A},W_{B})$ into the estimation $\hat{Y}_{\Delta}^{n}$.
       Similarly, Bob decodes $M_{A}$, $Y^{n}$, and the reliability
       information $(W_{A},W_{B})$, into the estimation
       $\hat{X}_{\Delta}^{n}$.
 \item Let $\mathcal{F}$ be a set of two-universal hash function
       \cite{carter:79} (see also Appendix \ref{u-hash}) from $\{ 0,1
       \}^{n}\times \{ 0,1 \}^{n}$ to $\mathcal{S}$.  Alice randomly
       choose a hash function $f\in \mathcal{F}$, and publicly tells the
       choice to Bob.  Then, Alice and Bob's final keys are
       $S=f(X_{\Delta}^{n},\hat{Y}_{\Delta}^{n})$ and
       $S'=f(\hat{X}_{\Delta}^{n},Y_{\Delta}^{n})$ respectively.
\end{enumerate}

In order to guarantee that Alice and Bob can compute the same key in
step 6, we set the rate $\frac{1}{n}\log|\mathcal{M}_{A}|$ and
$\frac{1}{n}\log|\mathcal{M}_{B}|$ of public messages according to the
following lemma, which is derived by modifying ``Slepian-Wolf Coding''
\cite{slepian:73} for continuous random variables.

\begin{lemma}
Suppose that we set
\begin{equation}
\frac{1}{n}\log|\mathcal{M}_{A}| >
H(X_{\Delta}|YW_{A}W_{B})\label{eq:compression1}
\end{equation}
and
\begin{equation}
\frac{1}{n}\log|\mathcal{M}_{B}| >
H(Y_{\Delta}|XW_{A}W_{B}),\label{eq:compression2}
\end{equation}
then there exist encoders and decoders such that the decoding error
probabilities $\Pr\{\hat{X}_{\Delta}^{n}\neq X_{\Delta}^{n}\}$ and
$\Pr\{\hat{Y}_{\Delta}^{n}\neq Y_{\Delta}^{n}\}$ tend to $0$ as $n \to
\infty$.

\end{lemma}
Thus, Eq.~(\ref{eq:agreement}) is satisfied for sufficiently large $n$.

In order to guarantee the security of the protocol, we set the key rate
$\frac{1}{n}\log|S|$ according to the following lemma, which is derived
by modifying the so-called ``left over hash lemma''
\cite{bennett:88,impagliazzo:89,bennett:95} for continuous random
variables.

\begin{lemma}
Suppose that we set
\begin{equation}
\frac{1}{n}\log|S|<
H(X_{\Delta}Y_{\Delta}|ZW_{A}W_{B})-\frac{1}{n}\log|\mathcal{M}_{A}||\mathcal{M}_{B}|,\label{eq:lower
bound}
\end{equation}
then
\begin{equation}
\label{eq:security2} H(S|Z^{n}W_{A}^{n}W_{B}^{n}M_{A}M_{B}F)\ge
\log|S|-\epsilon
\end{equation}
 is satisfied for sufficiently large $n$.
\end{lemma}
Note that $F$ is a random variable on ${\cal F}$, and all information
$C$ over the public channel correspond to
$(W_{A}^{n},W_{B}^{n},M_{A},M_{B},F)$ in this case.

From Eqs.~(\ref{eq:compression1})--(\ref{eq:lower bound}), we obtain the
following theorem that gives a lower bound on secret key rate
$R_{S}(V_{A},V_{B},V_{E})$ in this protocol.  \setcounter{theorem}{0}
\begin{theorem}
\label{theorem_lb} By using our proposed protocol, we achieve the lower
 bound on the secret key rate $R_{S}(V_{A},V_{B},V_{E})$ as
\begin{eqnarray}
\lefteqn{R_{S}(V_{A},V_{B},V_{E})}\notag\\ &\geq&
 H(X_{\Delta}Y_{\Delta}|ZW_{A}W_{B})
 -H(X_{\Delta}|YW_{A}W_{B}) \notag \\
&& \hspace{34mm}  -H(Y_{\Delta}|XW_{A}W_{B}).\label{eq:rate}
\end{eqnarray}
\end{theorem}
Note that from the rule in
Eqs.~(\ref{detemination:x})--(\ref{detemination:y}).  ,we can rewrite
the Eq.~(\ref{eq:rate}) as
\begin{eqnarray*}
\lefteqn{H(X_{\Delta}Y_{\Delta}|ZW_{A}W_{B})-H(X_{\Delta}|YW_{A}W_{B}) }
 \notag \\
&& \hspace{30mm} -H(Y_{\Delta}|XW_{A}W_{B}) \notag\\
&=&\sum_{w_{A},w_{B}}P_{W_{A}W_{B}}(w_{A},w_{B})\notag\\ &&\times
\max\{0, H(\Tilde{X}_{\Delta}|Z,W_{A}=w_{A},W_{B}=w_{B}) \\
&& ~~~~~~~~~~~~-H(\Tilde{X}_{\Delta}|Y,W_{A}=w_{A},W_{B}=w_{B}),\notag\\
&&\mspace{70mu}H(\Tilde{Y}_{\Delta}|Z,W_{A}=w_{A},W_{B}=w_{B}) \\
&& ~~~~~~~~~~~~-H(\Tilde{Y}_{\Delta}|X,W_{A}=w_{A},W_{B}=w_{B})\}\notag.
\end{eqnarray*}
For fixed $(W_{A},W_{B})=(w_{A},w_{B})$,
$H(\Tilde{X}_{\Delta}|Z,W_{A}=w_{A},W_{B}=w_{B})
-H(\Tilde{X}_{\Delta}|Y,W_{A}=w_{A},W_{B}=w_{B})$ is lower bound on the
secret key rate when we use only $\Tilde{X}_{\Delta}^{n}$ for generating
a secret key, $H(\Tilde{Y}_{\Delta}|Z,W_{A}=w_{A},W_{B}=w_{B})
-H(\Tilde{Y}_{\Delta}|X,W_{A}=w_{A},W_{B}=w_{B})$ is lower bound on the
secret key rate when we use only $\Tilde{Y}_{\Delta}^{n}$ for generating
a secret key, and $0$ is trivial lower bound on the secret key.  By the
rule in Eqs.~(\ref{detemination:x})--(\ref{detemination:y}), we choose
the maximum among these lower bounds on secret key rate for each
$(w_{A},w_{B})$ in order to make the lower bound on the secret key rate
as high as possible.

Note that encoding in step 5 and decoding in step 6 are implementable by
using low-density parity check codes \cite{muramatsu:05,coleman:06}.

\section{Comparison to a Protocol with Hard-Decision}

In this section, we will show the relation between signal-to-noise ratio
(SNR) and the key rate achieved by our proposed protocol for several
noise-to-noise ratio (NNR).  We will also show the comparisons between
the key rate achieved by our proposed protocol and the key rate achieved
by the protocol that Alice and Bob use only hard-decision for generating
a secret key.

\begin{figure}
 \centering \includegraphics[width=\linewidth]{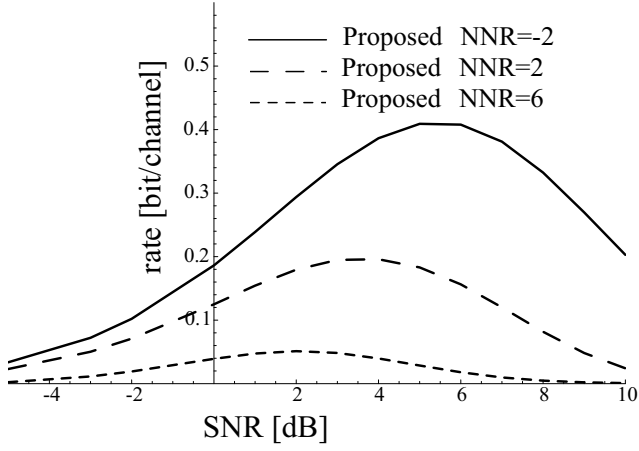} \caption{
 The relation between SNR and the key rate in our proposed protocol for
 several NNR.  } \label{SNR_NNR}
\end{figure}

The relation between (SNR) and the key rate achieved by our proposed
protocol for several NNR is presented in Fig.~\ref{SNR_NNR}, where sets
$E_{1}$, $E_{2}$, and $E_{3}$ are determined from fixed
$a_{1}=\frac{1}{3},a_{2}=\frac{2}{3},a_{3}=1$ in our proposed protocol.
Note that SNR is defined as $\frac{1}{V_{A}}$ and NNR is defined as
$\frac{V_{E}}{V_{B}}$, and we assume $V_{A}=V_{B}$.  From this figure,
we observe that we do not obtain a high key rate when SNR is too high or
too low.

\begin{figure} 
 \centering \subfigure[SNR$=1\text{[dB]}$]{ \label{vs_BSC1}
 \includegraphics[width=\linewidth]{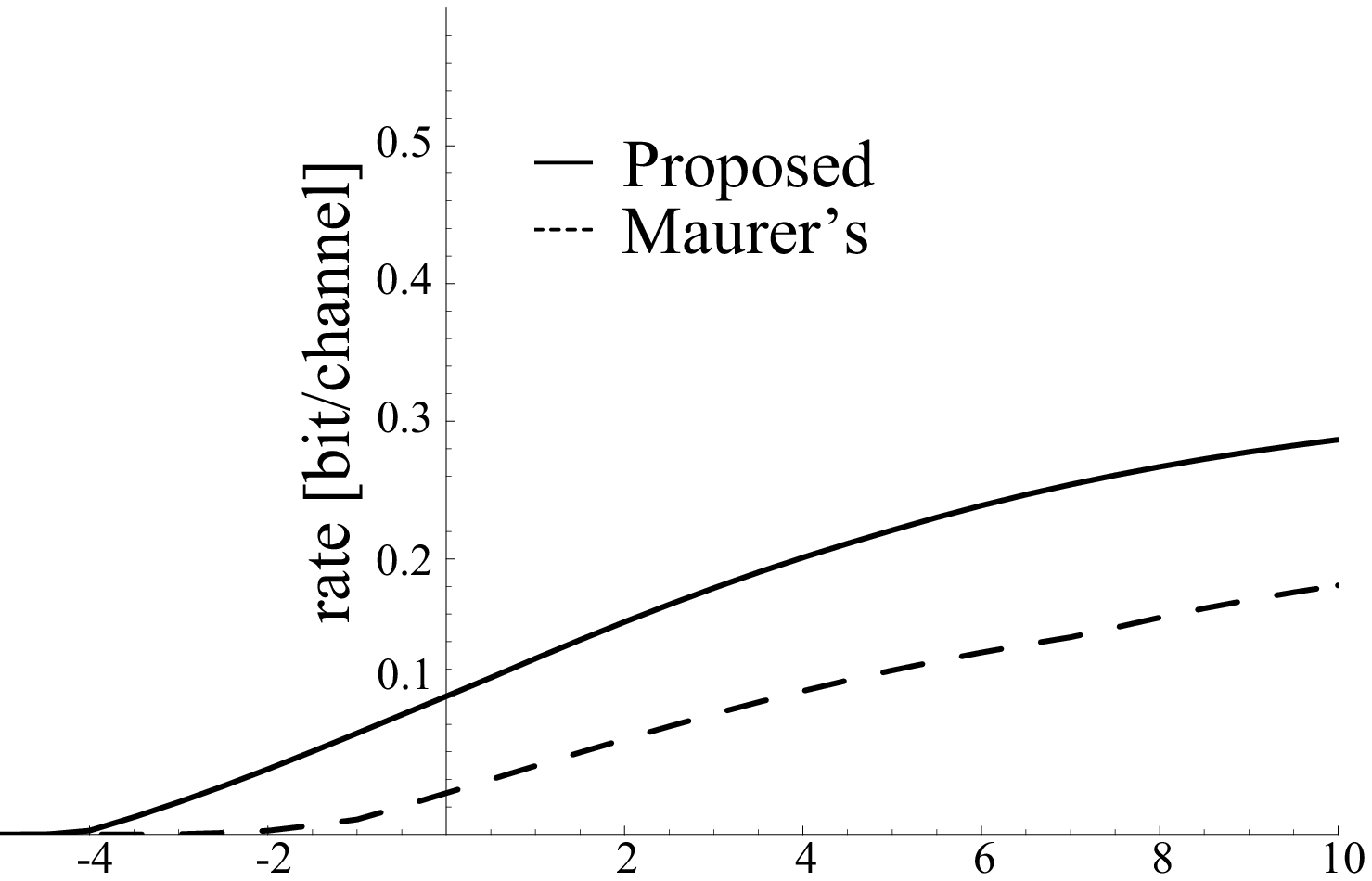}} \subfigure[SNR$=5\text{[dB]}$]{
 \label{vs_BSC2} \includegraphics[width=\linewidth]{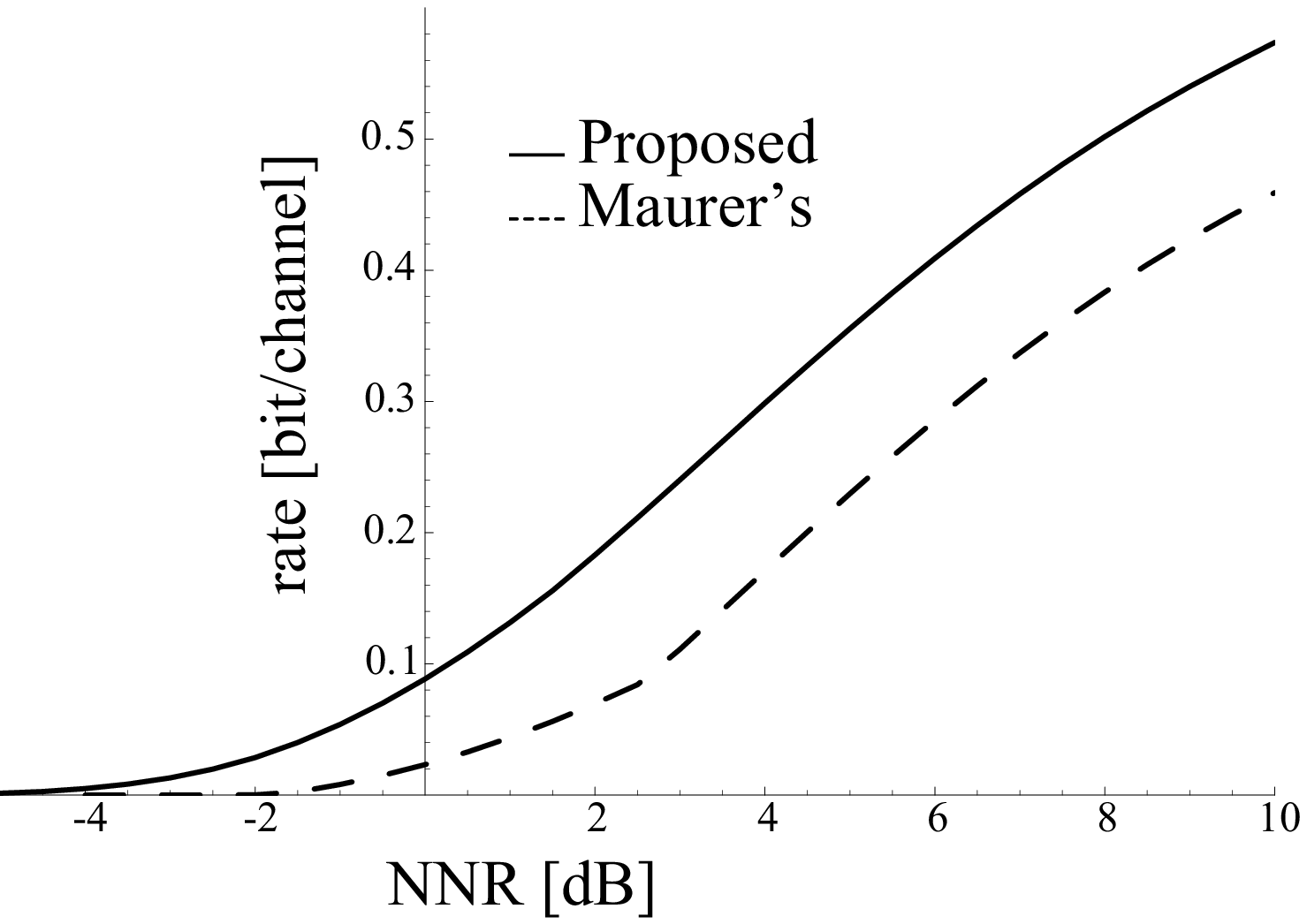}}
 \subfigure[SNR$=7\text{[dB]}$]{ \label{vs_BSC3}
 \includegraphics[width=\linewidth]{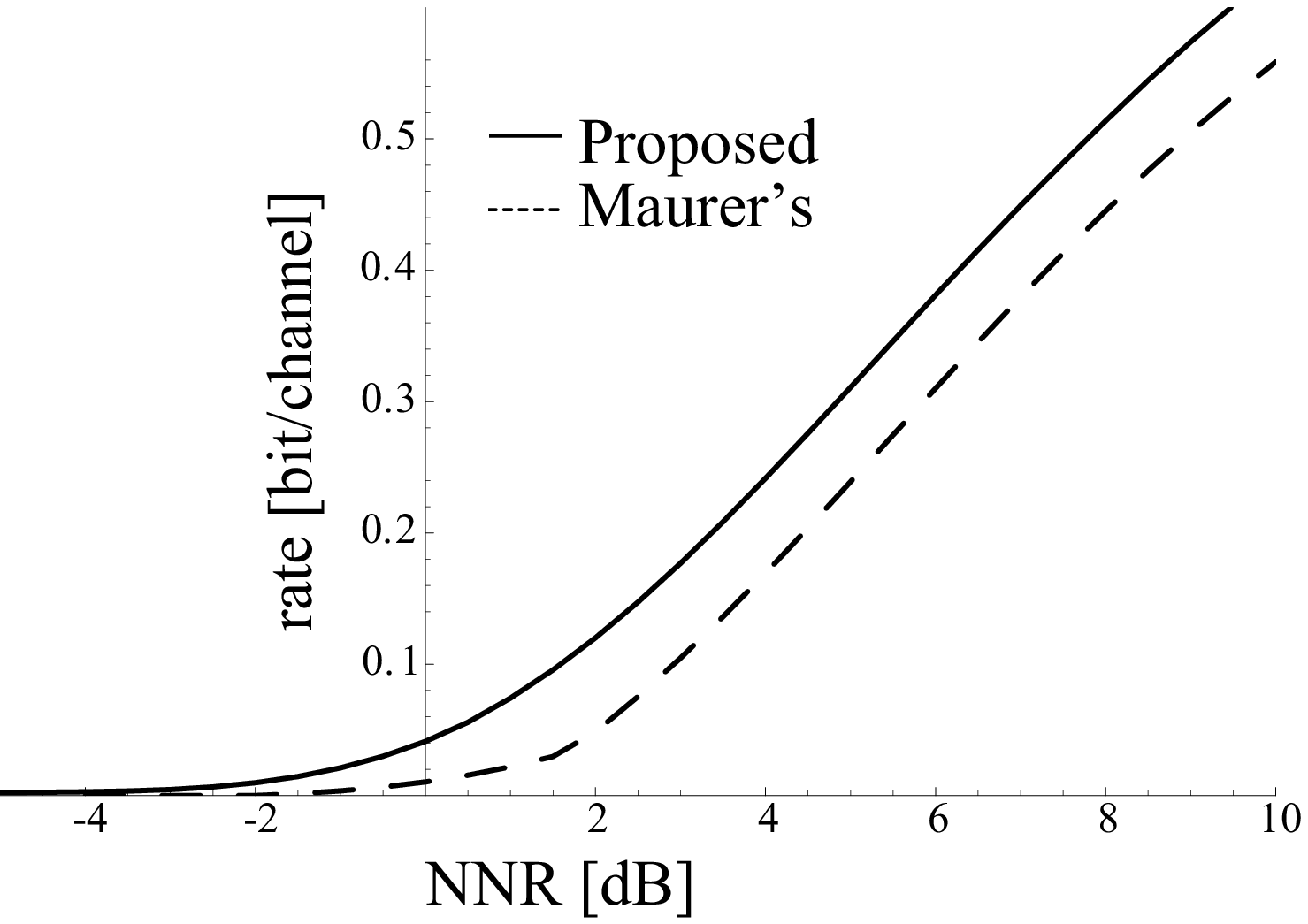}} \label{vs_BSC} \caption{
 The key rates achieved by our proposed protocol and Maurer's protocol.
 }
\end{figure}

In order to show advantage to use soft-decision, we compare the key rate
achieved by our proposed protocol and the key rate achieved by Maurer's
protocol in which Alice and Bob use only hard-decision for generating a
secret key.  The result of this comparison is presented in
Figs.~\ref{vs_BSC1}--\ref{vs_BSC3}.  In this comparison sets $E_{1}$,
$E_{2}$, and $E_{3}$ are determined from fixed
$a_{1}=\frac{1}{3},a_{2}=\frac{2}{3},a_{3}=1$ in our proposed protocol,
and the block length of repetition code used in Maurer's protocol is
optimally selected from $1$ to $10$ for each NNR.  From these figures,
we observe that we obtain a larger key rate by our proposed protocol
than by Maurer's protocol with all value of NNR.  Note that in Gaussian
Maurer's model, we should calculate the key rate by Maurer's protocol
for Eve who can use continuous random variables $Z^{n}$ to guess the
secret key.  However, the numerical calculation of the key rate by
Maurer's protocol in Gaussian Maurer's model is difficult when the block
length of repetition code used in his protocol is $2$ or larger.  Thus,
we calculate the key rate in BSC Maurer's model instead of Gaussian
Maurer's model when the block length of repetition code used in his
protocol is $2$ or larger.  In the calculation of the key rate in BSC
Maurer's model, we consider the situation that Alice, Bob, and Eve
hard-detect received signals according to the similar rule as in
Eqs. (\ref{hard-decision1}) and (\ref{hard-decision2}).  In this
situation, we can convert three Gaussian channels into independent
binary symmetric channels with error probabilities
$\epsilon_{A},\epsilon_{B},\epsilon_{E}$ given by
\begin{align}
\epsilon_{A}&=\frac{1}{2}\mathit{erfc}\Big(\sqrt{\frac{1}{V_{A}}}\Big),&
\epsilon_{B}&=\frac{1}{2}\mathit{erfc}\Big(\sqrt{\frac{1}{V_{B}}}\Big),&\notag\\
\epsilon_{E}&=\frac{1}{2}\mathit{erfc}\Big(\sqrt{\frac{1}{V_{E}}}\Big),
\end{align}
where the complementary error function $\mathit{erfc}(z)$ is defined as
\begin{equation}
\mathit{erfc}(z)=\frac{2}{\sqrt{\pi}}\int_{z}^{\infty}e^{-t^{2}}.
\end{equation}
Note that this way of the comparison gives Maurer's protocol advantage
because a wire-tapper in Gaussian Maurer's model is more powerful than
in BSC Maurer's model.\footnote{The wire-tapper in BSC Maurer's model
can use continuous random variables $Z^{n}$ to guess the secret key, but
one in BSC Maurer's model can only use quantized versions of them.}
Hence, the key rate achieved by Maurer's protocol in Gaussian Maurer's
model is lower than that presented in
Figs.~\ref{vs_BSC1}--\ref{vs_BSC3}.  

\section{Conclusion}

In this paper, we have proposed Gaussian Maurer's model and the protocol
with reliability information based on the result of the soft-decision in
this model.  As a result, we have obtained a higher key rate than
Maurer's protocol.  This is because that the correlation between
$X_{\Delta}$ in Eq.~(\ref{detemination:x}) and $Y$ and between
$Y_{\Delta}$ in Eq.~(\ref{detemination:y}) and $X$ obtained by using the
reliability information is stronger than the correlation between
$\Tilde{X}_{\Delta}$ in Eq.~(\ref{hard-decision1}) and
$\Tilde{Y}_{\Delta}$ in Eq.~(\ref{hard-decision2}) obtained by using the
hard-decision.

However, we do not know the optimal way to determine sets
$E_{1},\dots,E_{K}$ and its number $K$.  Intuitively, one may think that
the more sets we use, the higher rate we obtain.  However, this
intuition does not seem to be always true.  Actually, there exists the
case that we cannot obtain higher key rate though we use many sets.
Furthermore, we have to find the optimal signal constellation used by
the satellite.  These problems are future research agenda.

\section*{Acknowledgments}

We would like to thank Dr.~Jun Muramatsu for 
valuable discussions. This research also partly supported
by the Japan Society for the Promotion of Science under
Grants-in-Aid No.~00197137.

\appendix

\section{Proof of lemma 1}

\label{proof_lemma1} We only prove that if we set the rate
$\frac{1}{n}\log|\mathcal{M}_{A}|$ of public message according to
Eq.~(\ref{eq:compression1}), then there exist encoders and decoders such
that the decoding error probabilities $\Pr\{\hat{X}_{\Delta}^{n}\neq
X_{\Delta}^{n}\}$ tends to $0$ as $n\to \infty$.  The proof for the rate
$\frac{1}{n}\log|\mathcal{M}_{B}|$ of public message follows by
symmetry.

We use the so-called ``bin coding'' proposed by Cover \cite{cover} in
this proof.  The procedures of bin coding is as follows.

Assign every $x_{\Delta}^{n} \in \mathcal{X}_{\Delta}^{n}$ to one of
$|\mathcal{M}_{A}|$ bins independently according to the uniform
distribution on $\mathcal{M}_{A}$.

Alice sends the index $i$ of the bin to which $x_{\Delta}^{n}$ belongs.
Then let $\Bar{\varphi}_{n}(x_{\Delta}^{n})=i$.

For each $(y^{n},\mathbf{w}^{n})$, we define the set
$S_{n}(y^{n},\mathbf{w}^{n})\subset \mathcal{X}_{\Delta}^{n}$ as
\begin{eqnarray*}
\lefteqn{ S_{n}(y^{n},\mathbf{w}^{n}) } \\
&:=& \biggl\{ x_{\Delta}^{n} :
\frac{1}{n}\log
\frac{1}{P_{X_{\Delta}^{n}|Y^{n},\mathbf{W}^{n}}(x_{\Delta}^{n}|y^{n},\mathbf{w}^{n})} \\
&& ~~~~~~~~~~~~~~~~~\leq H(X_{\Delta}|Y\mathbf{W}) + \gamma \biggr\},\notag
\end{eqnarray*}
where $\gamma>0$ is an arbitrary fixed small constant, and we denote the
pair $(W_{A}^{n},W_{B}^{n})$ as $\mathbf{W}^{n}$.  Then, for given
$y^{n}$, $\mathbf{w}^{n}$, and the received index $i$, declare
$\Bar{\psi}_{n}(i,y^{n},\mathbf{w}^{n})=x_{\Delta}^{n}$ if there is one
and only one pair $(x_{\Delta}^{n},y^{n},\mathbf{w}^{n})$ such that
$\Bar{\varphi_{n}}(x_{\Delta}^{n})=i$ and ${x}_{\Delta}^{n} \in
S_{n}(y^{n},\mathbf{w}^{n})$.  Otherwise, declare an error.

We will evaluate the decoding error probability averaged over randomly
chosen encoders as follows.  We have an error if $X_{\Delta}^{n}$ is not
in $S_{n}(Y^{n},\mathbf{W}^{n})$ or if there is another symbol
${\hat{x}}_{\Delta}^{n} \in \mathcal{X}_{\Delta}^{n}$ in the same bin.
Thus, we can define the events of error
\begin{eqnarray*}
E_{n}^{(0)} &:=& \{ X_{\Delta}^{n} \notin S_{n}(Y^{n},\mathbf{W}^{n})
\},\notag\\ 
E_{n}^{(1)} &:=& \bigl\{ \exists {\hat{x}}^{n} \neq
X_{\Delta}^{n}:
\Bar{\varphi_{n}}({\hat{x}}_{\Delta}^{n})=\Bar{\varphi_{n}}(X_{\Delta}^{n})
\; \\
&& ~~~~~~~ \mbox{and} \; {\hat{x}}_{\Delta}^{n} \in S_{n}(Y^{n},\mathbf{W}^{n})
\bigr\},\notag
\end{eqnarray*}
Then the decoding error probability averaged over randomly chosen
encoders $\Pr\{X_{\Delta}^{n} \neq
\Bar{\psi_{n}}(\Bar{\varphi_{n}}(X_{\Delta}^{n}),Y^{n},\mathbf{W}^{n})\}$
is upper bounded as
\begin{eqnarray}
\lefteqn{ \Pr\{X_{\Delta}^{n}  \neq 
                  \Bar{\psi_{n}}(\Bar{\varphi_{n}}(X_{\Delta}^{n}),Y^{n},\mathbf{W}^{n})\}}
                  \notag \\
                  &=&\Pr\{ E_{n}^{(0)} \cup E_{n}^{(1)} \}\notag \\
                  &\le& \Pr\{E_{n}^{(0)}\} + \Pr\{E_{n}^{(1)}\}
                  \label{eq:error}.
\end{eqnarray}
$\Pr\{E_{n}^{(0)}\}$ is evaluated as
\begin{eqnarray}
\Pr\{E_{n}^{(0)}\} &=& \Pr\{ X_{\Delta}^{n} \notin
S_{n}(Y^{n},\mathbf{W}^{n}) \}\notag\\ &=& \Pr\biggl\{\frac{1}{n}\log
\frac{1}{P_{X_{\Delta}^{n}|Y^{n}\mathbf{W}^{n}}(X_{\Delta}^{n}|Y^{n}\mathbf{W}^{n})}
\notag \\
&& ~~~~~~> H(X_{\Delta}|Y\mathbf{W}) + \gamma \biggr \}\notag\\ &=&
\Pr\biggl\{\frac{1}{n} \sum_{i=1}^{n} \log
\frac{1}{P_{X_{\Delta}|Y\mathbf{W}}
(X_{\Delta}^{(i)}|Y^{(i)}\mathbf{W}^{(i)})} \notag \\
&& ~~~~~~> H(X_{\Delta}|Y\mathbf{W}) +
\gamma \biggr \} \label{e0},
\end{eqnarray}
which tends to $0$ as $n \to \infty$ by the weak law of large numbers.
To bound $\Pr\{E_{n}^{(1)}\}$, we rewrite it as
\begin{eqnarray}
\lefteqn{ \Pr\{E_{n}^{(1)}\}} \notag\\ &=& \Pr\bigl\{ \exists
{\hat{x}}_{\Delta}^{n} \neq X_{\Delta}^{n}:
\Bar{\varphi}_{n}({\hat{x}}_{\Delta}^{n})=\Bar{\varphi}_{n}(X_{\Delta}^{n})\;
\notag \\
&& ~~~~~~~\mbox{and} \;{\hat{x}}_{\Delta}^{n} \in S_{n}(Y^{n},\mathbf{W}^{n}) \bigr\}
\notag\\ &=& \int_{\mathcal{Y}^{n}}p_{Y^{n}}(y^{n})
\sum_{(x_{\Delta}^{n},\mathbf{w}^{n})\in \mathcal{X}_{\Delta}^{n} \times
\mathcal{W}_{A}^{n} \times \mathcal{W}_{B}^{n}} \notag \\
&& P_{X_{\Delta}^{n}\mathbf{W}^{n}|y^{n}}(x_{\Delta}^{n},\mathbf{w}^{n})
g_{n}(x_{\Delta}^{n},y^{n},\mathbf{w}^{n}) \,dy^{n},
\label{e11}
\end{eqnarray}
where
\begin{eqnarray}
\lefteqn{g_{n}(x_{\Delta}^{n},y^{n},\mathbf{w}^{n})}\notag\\
&=&\Pr\bigl\{ \exists {\hat{x}}_{\Delta}^{n} \neq x_{\Delta}^{n}:
\Bar{\varphi}_{n}({\hat{x}}_{\Delta}^{n})=\Bar{\varphi}_{n}(x_{\Delta}^{n})\:
\notag \\
&& ~~~~\mbox{and} \; ({\hat{x}}_{\Delta}^{n}) \in S_{n}(y^{n},\mathbf{w}^{n})
\bigr\}. \label{define:g}
\end{eqnarray}
Furthermore, we can rewrite (\ref{define:g}) as
\begin{align}
g_{n}(x_{\Delta}^{n},y^{n},\mathbf{w}^{n})
&=\sum_{\genfrac{}{}{0pt}{2}{{\hat{x}}_{\Delta}^{n} \neq
x_{\Delta}^{n}}{{\hat{x}}_{\Delta}^{n} \in S_{n}(y^{n},\mathbf{w}^{n})}}
\Pr\{
\Bar{\varphi}_{n}({\hat{x}}_{\Delta}^{n})=\Bar{\varphi}_{n}(x_{\Delta}^{n})
\} \notag\\ &= \sum_{\genfrac{}{}{0pt}{2}{{\hat{x}}_{\Delta}^{n} \neq
x_{\Delta}^{n}}{ {\hat{x}}_{\Delta}^{n} \in
S_{n}(y^{n},\mathbf{w}^{n})}} \frac{1}{|\mathcal{M}_{A}|} \notag\\ &\leq
\sum_{{\hat{x}}_{\Delta}^{n} \in S_{n}(y^{n},\mathbf{w}^{n})}
\frac{1}{|\mathcal{M}_{A}|} \notag\\
&=\frac{|S_{n}(y^{n},\mathbf{w}^{n})|}{|\mathcal{M}_{A}|}\label{e13}
\end{align}
If ${\hat{x}}_{\Delta}^{n} \in S_{n}(y^{n},\mathbf{w}^{n})$, then from
the definition of $S_{n}(y^{n},\mathbf{w}^{n})$, we have
\begin{equation}
P_{X_{\Delta}^{n}|y^{n},\mathbf{w}^{n}}({\hat{x}}_{\Delta}^{n}) \geq
2^{-n(H(X_{\Delta}|Y\mathbf{W}) + \gamma)}.\notag
\end{equation}
Thus, we have
\begin{eqnarray}
1 &\geq& \sum_{{\hat{x}}_{\Delta}^{n} \in S_{n}(y^{n},\mathbf{w}^{n})}
       P_{X_{\Delta}^{n}|Y^{n}\mathbf{W}^{n}}(x_{\Delta}^{n}|y^{n},\mathbf{w}^{n})
       \notag \\
       &\geq& |S_{n}(y^{n},\mathbf{w}^{n})|2^{-n(H(X_{\Delta}|Y\mathbf{W})
       + \gamma)}.\notag
\end{eqnarray}
Hence, we have
\begin{equation}
|S_{n}(y^{n},\mathbf{w}^{n})|\leq 2^{n(H(X_{\Delta}|Y\mathbf{W}) +
 \gamma)}.\label{e14}
\end{equation}
From Eqs.(\ref{e11})--(\ref{e14}), we upper bound $\Pr\{E_{n}^{(1)}\}$
as
\begin{eqnarray}
\Pr \{E_{n}^{(1)}\} &\leq& \int_{\mathcal{Y}^{n}}p_{Y^{n}}(y^{n})
\sum_{(x_{\Delta}^{n},\mathbf{w}^{n})\in \mathcal{X}_{\Delta}^{n} \times
\mathcal{W}_{A}^{n} \times \mathcal{W}_{B}^{n}} \notag \\
&& P_{X_{\Delta}^{n}\mathbf{W}^{n}|y^{n}}(x_{\Delta}^{n},\mathbf{w}^{n}) 
\frac{2^{n(H(X_{\Delta}|Y\mathbf{W}) +
\gamma)}}{|\mathcal{M}_{A}|}\,dy^{n}\notag\\ &\leq&
\frac{2^{n(H(X_{\Delta}|Y\mathbf{W}) +
\gamma)}}{|\mathcal{M}_{A}|}\notag\\
&=&2^{-\log|\mathcal{M}_{A}|}2^{n(H(X_{\Delta}|Y\mathbf{W}) + \gamma)},
\end{eqnarray}
which exponentially tends to $0$ as $n \to \infty$ if
$\frac{1}{n}\log|\mathcal{M}_{A}|>H(X_{\Delta}|Y\mathbf{W})+\gamma$.

Since the decoding error probability $\Pr\{X_{\Delta}^{n} \neq
\Bar{\psi_{n}}(\Bar{\varphi_{n}}(X_{\Delta}^{n}),Y^{n},\mathbf{W}^{n})\}$
of randomly chosen code tends to $0$ as $n \to \infty$, there exist at
least one pair of an encoder and a decoder such that the decoding error
probability $\Pr\{{\hat{X}}_{\Delta}^{n} \neq X_{\Delta}^{n}\}$ tends to
$0$ as $n\to \infty$.

\section{Proof of lemma 2}

\label{proof_lemma2} \indent In this Appendix, we will show the proof of
lemma 2.  In section \ref{u-hash}, we introduce a two-universal hash
family, which is used for computation of a secret key.  In section
\ref{relate-var}, we define the security of the protocol in the sense of
the variational distance, and we show the relation between the security
of the protocol in the sense of the variational distance and the
condition Eq.~(\ref{eq:security}).  This relation implies that if the
security of the protocol in the sense of the variational distance is
satisfied, then the condition Eq.~(\ref{eq:security}) is satisfied.  In
section \ref{size-security}, we relate the size $|{\cal S}|$ of a secret
key $S$ and the size $|{\cal M}_{A}\times {\cal M}_{B}|$ of public
messages $\mathbf{M}=(M_{A},M_{B})$ to the security of the protocol, and
we show that if we set $\frac{1}{n}\ln|{\cal S}| <
H(X_{\Delta}Y_{\Delta}|ZW_{A}W_{B}) - \frac{1}{n} \ln|{\cal
M}_{A}\times{\cal M}_{B}|$, then there exists at least one hash function
$f$ that satisfy Eq.~(\ref{eq:security}) for sufficiently large $n$.

For the simplicity of notation, we denotes the integral over
$\mathbb{R}^{n}$ as $\int$ unless otherwise specified, and we
abbreviates
$P_{\mathbf{R}^{n}\mathbf{M}^{n}|Z^{n}\mathbf{W}^{n}}(\cdot,\cdot|z^{n},\mathbf{w}^{n})$
as $P_{\mathbf{R}^{n}\mathbf{M}^{n}|z^{n},\mathbf{w}^{n}}(\cdot,\cdot)$.
The variational distance $\| P_{1}-P_{2} \|$ between the probability
distribution $P_1$ and $P_2$ on $\mathcal{V}$ is defined as
\begin{eqnarray}
\| P_{1} - P_{2} \| := \sum_{v \in {\cal V}} | P_{1}(v) - P_{2}(v) |.
\end{eqnarray}

\subsection{two-universal hash family}
\label{u-hash}

In order to extract an almost secret string (secret key $S$) from a
partially secret strings (a pair $\mathbf{R}^{n}$ of random variables
$X_{\Delta}^{n}$ and $Y_{\Delta}^{n}$), we use a two-universal hash
family $\mathcal{F}$.  A set $\mathcal{F}$ of functions $f:
\mathcal{X}_{\Delta}^{n}\times \mathcal{Y}_{\Delta}^{n} \to \mathcal S$
is said to be a \textit{two-universal hash family} if we have
\begin{eqnarray}
P_F\left(\{ f \in \mathcal{F} \mid f(\mathbf{r}^{n}) =
	f({\mathbf{r}^\prime}^{n}) \}\right) \le \frac{1}{|{\cal S}|}~
	\label{eq-two-universal}
\end{eqnarray}
for any $\mathbf{r}^{n}\neq{\mathbf{r}^{\prime}}^{n} \in
\mathcal{X}_{\Delta}^{n}\times \mathcal{Y}_{\Delta}^{n}$, where $F$
denotes a random variable on ${\cal F}$ and $P_F$ denotes the uniform
distribution on ${\cal F}$.  For given Eve's received signals $z^{n} \in
\mathbb{R}^{n}$ and reliability information $\mathbf{w}^{n} \in {\cal
W}_{A}\times{\cal W}_{B}$, the jointly conditional distribution
$P_{S\mathbf{M}|z^{n},\mathbf{w}^{n}}(s,\mathbf{m})$ of a secret key
$S=f(\mathbf{R}^{n})$ and public message $\mathbf{M}$ is given by
\begin{eqnarray*}
P_{S\mathbf{M}|z^{n},\mathbf{w}^{n}}(s,\mathbf{m}) &:=&
\sum_{\mathbf{r}^{n} \in f^{-1}(s)}
P_{\mathbf{R}^{n}\mathbf{M}|z^{n},\mathbf{w}^{n}}(\mathbf{r}^{n},\mathbf{m})\notag\\
&=&
P_{\mathbf{R}^{n}\mathbf{M}|z^{n},\mathbf{w}^{n}}(f^{-1}(s),\mathbf{m}),
\end{eqnarray*}
where $f^{-1}(s) := \{ \mathbf{r}^{n} \in
\mathcal{X}_{\Delta}^{n}\times\mathcal{Y}_{\Delta}^{n} \mid
f(\mathbf{r}^{n}) = s \}$ is the subset of a set
$\mathcal{X}_{\Delta}^{n} \times \mathcal{Y}_{\Delta}^{n}$ such that
$f(\mathbf{r}^{n})=s$.  Note that since $S$ depends on a hash function
$f$, it should be referred as $S_f$.  But, we use the above notation for
convenience in this paper.

\subsection{The security of the protocol in the sense of the variational
  distance}

\label{relate-var}

In order to prove lemma 2, we define the security of the protocol in the
sense of the variational distance in this section.  If a secret key $S$
is independent of Eve's information and its distribution $P_{S}$ is
close to the uniform distribution $P_{\bar{S}}$ on ${\cal S}$, we decide
that the secret key $S$ is secure in the sense of the variational
distance.  In the other words, we define the security of the protocol as
\begin{eqnarray}
\label{definition-of-security} 
\Delta_f &:=& \int p_{Z^{n}}(z^{n})
\sum_{\mathbf{w}^{n}\in{\cal W}_{A}^{n}\times{\cal W}_{B}^{n}}
P_{\mathbf{W}^{n}|z^{n}}(\mathbf{w}^{n}) \notag \\
&& \| P_{S\mathbf{M}|z^{n},\mathbf{w}^{n}} -P_{\bar{S}} \times
P_{\mathbf{M}|z^{n},\mathbf{w}^{n}} \| dz^{n},
\end{eqnarray}
where $P_{\mathbf{M}|z^{n},\mathbf{w}^{n}}$ is the marginal distribution
of $P_{S\mathbf{M}|z^{n},\mathbf{w}^{n}}$, and $P_{\bar{S}} \times
P_{\mathbf{M}|z^{n},\mathbf{w}^{n}}$ is the product distribution of
$P_{\bar{S}}$ and $P_{\mathbf{M}|z^{n},\mathbf{w}^{n}}$

As an extension of \cite[Lemma 1]{csiszar:04} to continuous random
variable, the following lemma relates the security of the protocol in
the sense of the variational distance to the security of the protocol in
the sense of the entropy shown in Eq.~(\ref{eq:security}).

\setcounter{theorem}{2}
\begin{lemma}

The conditional entropy $H(S|Z^{n}\mathbf{W}^{n}\mathbf{M}F)$ is lower
bounded by
\begin{eqnarray}
H(S|Z^{n}\mathbf{W}^{n}\mathbf{M}F) &\ge& (1-\mathbb{E}_f[\Delta_{f}]) \ln
|{\cal S}| \notag \\
&& - \mathbb{E}_f[\Delta_f] \log
\frac{1}{\mathbb{E}_f[\Delta_f]}.
\end{eqnarray}
\end{lemma}
Note that since $\mathbf{W}^{n}=(W_{A}^{n},W_{B}^{n})$ and
$\mathbf{M}=(M_{A},M_{B})$, the conditional entropy
$H(S|Z^{n}\mathrm{W}^{n}\mathbf{M}F)$ equivalent to
$H(S|Z^{n}W_{A}^{n}W_{B}^{n}M_{A},M_{B}F)$ in Eq.~(\ref{eq:security2}).
From this lemma, if $\mathbb{E}_f[\Delta_f]$ is sufficiently small, a
secret key $S$ is secure in the sense of the entropy.

\begin{proof}
Let
\begin{eqnarray}
\Delta_{f,\mathbf{m},z^{n},\mathbf{w}^{n}} := \|
P_{S|\mathbf{m},z^{n},\mathbf{w}^{n}} - P_{\bar{S}} \|.
\end{eqnarray}
Then, we can rewrite $\Delta_{f}$ as
\begin{eqnarray}
\Delta_f &=& \int
p_{Z^{n}}(z^{n}) \sum_{\mathbf{m},\mathbf{w}^{n}} \nonumber \\
&& P_{\mathbf{M}\mathbf{W}^{n}|z^{n}}(\mathbf{m},\mathbf{w}^{n})
\Delta_{f,\mathbf{m},z^{n},\mathbf{w}^{n}}\, dz^{n}
	\label{eq-lemma-security-equivalent-1} 
\end{eqnarray}

For given $z^{n} \in \mathbb{R}^{n}$, $\mathbf{w}^{n} \in {\cal
W}_{A}^{n}\times{\cal W}_{B}^{n}$, and $\mathbf{m} \in {\cal
M}_{A}\times {\cal M}_{B}$, we obtain
\begin{eqnarray}
H(S|\mathbf{M} &=& \mathbf{m},
 Z^{n}=z^{n},\mathbf{W}^{n}=\mathbf{w}^{n},F=f) \notag \\
&\ge& \log |{\cal S}| -
 \Delta_{f,\mathbf{m},z^{n},\mathbf{w}^{n}} \log \frac{|{\cal
 S}|}{\Delta_{f,\mathbf{m},z^{n},\mathbf{w}^{n}}}, \notag \\
 \label{eq-lemma-security-equivalent-3}
\end{eqnarray}
which follows from the continuity of entropy \cite{cover} in the similar
 way as \cite[Lemma 1]{csiszar:04}.

The second term of Eq.~(\ref{eq-lemma-security-equivalent-3}) is upper
bonded as follow.  Since $t \log \frac{1}{t}$ is a concave function, we
obtain
\begin{eqnarray}
\label{eq-lemma-security-equivalent-4} \sum_{\mathbf{m},\mathbf{w}^{n}}
\lefteqn{  P_{\mathbf{M}\mathbf{W}^{n}|z^{n}}(\mathbf{m},\mathbf{w}^{n})
 \Delta_{f,\mathbf{m},z^{n},\mathbf{w}^{n}} \log \frac{|{\cal
 S}|}{\Delta_{f,\mathbf{m},z^{n},\mathbf{w}^{n}}} } \notag \\
&\le& \Delta_{f,z} \log
 \frac{|{\cal S}|}{\Delta_{f,z}}
\end{eqnarray}
from Jensen's inequality for $\mathbf{w}^{n},\mathbf{m}$, where we let
$\Delta_{f,z^{n}} := \sum_{\mathbf{m},\mathbf{w}^{n}}
P_{\mathbf{M}\mathbf{W}^{n}|z^{n}}(\mathbf{m},\mathbf{w}^{n})
\Delta_{f,\mathbf{m},z^{n},\mathbf{w}^{n}}.$ Averaging
Eq.~(\ref{eq-lemma-security-equivalent-4}) over $z^{n}$, we obtain
\begin{eqnarray}
\label{eq-lemma-security-equivalent-5} \int p_{Z^{n}}(z^{n})
\Delta_{f,z^{n}} \log \frac{|{\cal S}|}{\Delta_{f,z}}dz^{n}
 \le
\Delta_f \log \frac{|{\cal S}|}{\Delta_f}
\end{eqnarray}
from Jensen's inequality for $z^{n}$.  Moreover, averaging
Eq.~(\ref{eq-lemma-security-equivalent-5}) over $f$, we obtain
\begin{eqnarray}
\mathbb{E}_f\left[\Delta_f \log \frac{|{\cal S}|}{\Delta_f}\right] \le
\mathbb{E}_f[\Delta_f] \log \frac{|{\cal S}|}{\mathbb{E}_f[\Delta_f]}
\end{eqnarray}
from Jensen's inequality for $f$.
\end{proof}

Note that when we use Jensen's inequality for a continuous random
variable, the condition of absolutely integrable
\begin{eqnarray}
\int p_{Z^{n}}(z^{n}) | \Delta_{f,z^{n}} | dz^{n} < \infty
\end{eqnarray}
must be satisfied \cite{itou}.  In this case, from the fact that $0 \le
\Delta_{f,z^{n}} \le 2$, this condition is satisfied.

\subsection{The relation between the size of a secret key and the
  security of the protocol}
\label{size-security}
The following lemma relates the size $|{\cal S}|$ of a secret key $S$
and the size $|{\cal M}_{A}\times {\cal M}_{B}|$ of public messages
$\mathbf{M}$ to the security of the protocol.
\begin{lemma}
 \label{lemma-privacy-amplification} For the size $|{\cal S}|$ of a
secret key $S$, the size $|{\cal M}_{A}\times {\cal M}_{B}|$ of public
messages $\mathbf{M}$, and the security of the protocol $\Delta_{f}$, we
have
\begin{eqnarray}
 \label{eq-lemma-pa-0} \lefteqn{\mathbb{E}_f [ \Delta_f ]}\notag\\ &\le&
\sqrt{\frac{|{\cal S}| |{\cal M}_{A}\times{\cal M}_{B}|}{2^{\alpha n}}}
\notag\\ &&+ 2 \int p_Z^{n}(z^{n}) \sum_{\mathbf{w}^{n}}
P_{\mathbf{W}^{n}|z^{n}}(\mathbf{w}^{n})\notag\\ &&\times
P_{\mathbf{R}^{n}|z^{n}\mathbf{w}^{n}}\left( \left\{ \mathbf{r}^{n} \in
{\cal X}_{\Delta}^{n}\times{\cal Y}_{\Delta}^{n} \mid \right. \right. \notag \\
&& ~~~~ \left. \left. -\frac{1}{n}\log
P_{\mathbf{R}^{n}|z^{n}\mathbf{w}^{n}}(\mathbf{r}^{n}) < \alpha \right\}
\right) dz^{n},
\end{eqnarray}
where $\mathbb{E}_f$ denotes expectation for a uniform distribution on
${\cal F}$.
\end{lemma}

\begin{proof}
This proof is based on the techniques in \cite[Chapter 5]{renner:05b}.
In the following, we will prove
\begin{eqnarray}
\label{eq-lemma-pa-1} \lefteqn{\mathbb{E}_f [
\Delta_{f,z^{n},\mathbf{w}^{n}} ]}\notag\\ &\le& \sqrt{\frac{|{\cal S}|
|{\cal M}_{A}\times{\cal M}_{B}|}{2^{\alpha n}}}\notag \\ &&+ 2
P_{\mathbf{R}^{n}|z^{n}\mathbf{w}^{n}}\left( \left\{ \mathbf{r}^{n} \in
{\cal X}_{\Delta}^{n}\times{\cal Y}_{\Delta}^{n} \mid
                                             \right. \right. \notag \\
&& ~~~~~~ \left. \left. -\frac{1}{n}\log
P_{\mathbf{R}^{n}|z^{n}\mathbf{w}^{n}}(\mathbf{r}^{n}) < \alpha \right\}
\right),
\end{eqnarray}
where
\begin{eqnarray}
\label{eq-lemma-pa-1-1} \Delta_{f,z^{n},\mathbf{w}^{n}} = \|
P_{S\mathbf{M}|z^{n}\mathbf{w}^{n}} - P_{\bar{S}} \times
P_{\mathbf{M}|z^{n}\mathbf{w}^{n}} \|,
\end{eqnarray}
Averaging Eq.~(\ref{eq-lemma-pa-1}) over $z^{n}$ and $\mathbf{w}^{n}$,
we obtain Eq.~(\ref{eq-lemma-pa-0}).

For given $z^{n} \in \mathbb{R}^{n}$ and $\mathbf{w}^{n} \in {\cal
W}_{A}^{n}\times {\cal W}_{B}^{n}$, we define the set $A_{n} \subset
{\cal X}_{\Delta}^{n}\times {\cal Y}_{\Delta}^{n}$ as
\begin{eqnarray*}
\!\!\! A_{n} := \left\{ \mathbf{r}^{n} \in {\cal X}_{\Delta}^{n}\times {\cal
Y}_{\Delta}^{n} \mid  -\frac{1}{n}\log
P_{\mathbf{R}^{n}|z^{n}\mathbf{w}^{n}}(\mathbf{r}^{n}) \ge \alpha
\right\},
\end{eqnarray*}
and we define the set $A_{n}^c$ as the complement of $A_{n}$ on ${\cal
 X}_{\Delta}^{n}\times {\cal Y}_{\Delta}^{n}$.  Then,
 $\Delta_{f,z^{n},\mathbf{w}^{n}}$ for given $f\in {\cal F}$ is upper
 bounded by
\begin{eqnarray}
\lefteqn{ \| P_{S\mathbf{M}|z^{n}\mathbf{w}^{n}} - P_{\bar{S}} \times
P_{\mathbf{M}|z^{n}\mathbf{w}^{n}} \| } \notag \\ &=&
\sum_{s,\mathbf{m}} |
P_{\mathbf{R}^{n}\mathbf{M}|z^{n}\mathbf{w}^{n}}(f^{-1}(s), \mathbf{m})
\notag \\
&& - P_{\bar{S}}(s) P_{\mathbf{M}|z^{n}\mathbf{w}^{n}}(\mathbf{m}) |
\label{eq-lemma-pa-2} \\ 
&=& \sum_{s,\mathbf{m}} |
P_{\mathbf{R}^{n}\mathbf{M}|z^{n}\mathbf{w}^{n}}(f^{-1}(s)\cap
A_{n},\mathbf{m}) \notag \\
&& - P_{\bar{S}}(s)
P_{\mathbf{M}|z^{n}\mathbf{w}^{n}}(A_{n},\mathbf{m}) \notag\\
&& +
P_{\mathbf{R}^{n}\mathbf{M}|z^{n}\mathbf{w}^{n}}(f^{-1}(s)\cap
A_{n}^{c},\mathbf{m}) \notag \\
&& - P_{\bar{S}}(s)
P_{\mathbf{M}|z^{n}\mathbf{w}^{n}}(A_{n}^{c},\mathbf{m}) |
\label{eq-lemma-pa-3} \\ &\le& \sum_{s,\mathbf{m}} h_{n}(s,\mathbf{m})
+\sum_{s,\mathbf{m}}P_{\mathbf{R}^{n}\mathbf{M}|z^{n}\mathbf{w}^{n}}(f^{-1}(s)\cap
A_{n}^{c},\mathbf{m})\notag \\ &&+\sum_{s,\mathbf{m}} P_{\bar{S}}(s)
P_{\mathbf{R}^{n}\mathbf{M}|z^{n}\mathbf{w}^{n}}(A_{n}^{c},\mathbf{m})
\label{eq-lemma-pa-4} \\ &=& \sum_{s,\mathbf{m}} h_{n}(s,\mathbf{m}) +2
P_{\mathbf{R}^{n}|z^{n}\mathbf{w}^{n}}(A_{n}^{c}).
\label{eq-lemma-pa-5}
\end{eqnarray}
where
\begin{eqnarray}
h_{n}(s,\mathbf{m}) &=& |P_{\mathbf{R}^{n}\mathbf{M}|z^{n}\mathbf{w}^{n}}(f^{-1}(s)\cap
A_{n},\mathbf{m}) \notag \\
&& - P_{\bar{S}}(s)
P_{\mathbf{R}^{n}\mathbf{M}|z^{n}\mathbf{w}^{n}}(A_{n},\mathbf{m})|.
\end{eqnarray}
Eq.~(\ref{eq-lemma-pa-2}) follows from the definition of the variational
distance and $f^{-1}(s)$.  Eq.~(\ref{eq-lemma-pa-3}) follows from the
fact that $(f^{-1}(s) \cap A_{n}) \cap (f^{-1}(s) \cap A_{n}^c) =
\emptyset$, $~f^{-1}(s) = (f^{-1}(s) \cap A_{n}) \cup (f^{-1}(s) \cap
A_{n}^c)$, and $P_{\mathbf{M}|z^{n}\mathbf{w}^{n}}(\mathbf{m}) =
P_{\mathbf{R}^{n}\mathbf{M}|z^{n}\mathbf{w}^{n}}(A_{n},\mathbf{m}) +
P_{\mathbf{R}^{n}\mathbf{M}|z^{n}\mathbf{w}^{n}}(A_{n}^c,\mathbf{m})$.
Eq.~(\ref{eq-lemma-pa-4}) follows from the triangle inequality.
Eq.~(\ref{eq-lemma-pa-5}) follows from the fact that $\cup_{s \in {\cal
S}} f^{-1}(s) = {\cal X}^{n}_{\Delta}\times{\cal Y}^{n}_{\Delta}$.
 By
regarding the first term in Eq.~(\ref{eq-lemma-pa-5}) as an inner
product, and by using the Cauchy-Schwarz inequality, we can upper bound
the first term in Eq.~(\ref{eq-lemma-pa-5}) by
\begin{eqnarray}
\lefteqn{  \sum_{s,\mathbf{m}} h_{n}(s,\mathbf{m}) } \notag \\
&\le&\sqrt{|{\cal S}| |{\cal
M}_{A}\times{\cal M}_{B}| \sum_{s,\mathbf{m}}h_{n}(s,\mathbf{m})^{2}}
\label{eq-lemma-pa-6}
\end{eqnarray}

Furthermore, we can rewrite the inside of the root of
 Eq.~(\ref{eq-lemma-pa-6}) as
\begin{eqnarray}
\lefteqn{\sum_{s,\mathbf{m}} h_{n}(s,\mathbf{m})^2} \notag\\ &=&
\sum_{s,\mathbf{m}} \left\{
P_{\mathbf{R}^{n}\mathbf{M}|z^{n}\mathbf{w}^{n}}(f^{-1}(s)\cap
A_{n},\mathbf{m})^2 \right. \notag \\ 
&& \mspace{30mu}- 2
P_{\mathbf{R}^{n}\mathbf{M}|z^{n}\mathbf{w}^{n}}(f^{-1}(s)\cap
A_{n},\mathbf{m}) \notag \\
&& ~~~~~~~~~P_{\bar{S}}(s)
P_{\mathbf{R}^{n}\mathbf{M}|z^{n}\mathbf{w}^{n}}(A_{n},\mathbf{m})
\notag\\ && \left.  \mspace{30mu}+ P_{\bar{S}}(s)^2
P_{\mathbf{R}^{n}\mathbf{M}|z^{n}\mathbf{w}^{n}}(A_{n},\mathbf{m})^{2}
\right\} \notag\\ &=& \sum_{s,\mathbf{m}}
P_{\mathbf{R}^{n}\mathbf{M}|z^{n}\mathbf{w}^{n}}(f^{-1}(s)\cap
A_{n},\mathbf{m})^2 \notag \\
&& - \sum_{\mathbf{m}} \frac{1}{|{\cal S}|}
P_{\mathbf{R}^{n}\mathbf{M}|z^{n}\mathbf{w}^{n}}(A_{n},\mathbf{m})^2,\label{eq-lemma-pa-7}
\end{eqnarray}
where Eq.~(\ref{eq-lemma-pa-7}) follows from the fact that
$P_{\bar{S}}(s) = \frac{1}{|{\cal S}|}$ and $\sum_{s}
P_{\mathbf{R}^{n}\mathbf{M}|z^{n}\mathbf{w}^{n}}(f^{-1}(s)\cap
A_{n},\mathbf{m})
 =
P_{\mathbf{R}^{n}\mathbf{M}|z^{n}\mathbf{w}^{n}}(A_{n},\mathbf{m})$.
Then, we can rewrite the first term of Eq.~(\ref{eq-lemma-pa-7}) as
\begin{eqnarray}
\lefteqn{ \sum_{s,\mathbf{m}}
P_{\mathbf{R}^{n}\mathbf{M}|z^{n}\mathbf{w}^{n}}(f^{-1}(s)\cap
A_{n},\mathbf{m})^2}\notag \\ 
&=& \sum_{s,\mathbf{m}}
\sum_{\mathbf{r}^{n}, {\mathbf{r}^\prime}^{n} \in f^{-1}(s) \cap A_{n} }
P_{\mathbf{R}^{n}\mathbf{M}|z^{n}\mathbf{w}^{n}}(\mathbf{r}^{n},\mathbf{m})
\notag \\
&& ~~~~~~~~~~~~~~~~~~~~~~~~~~~P_{\mathbf{R}^{n}\mathbf{M}|z^{n}\mathbf{w}^{n}}({\mathbf{r}^\prime}^{n},\mathbf{m})
\notag \\ &=& \sum_{\mathbf{m}} \sum_{\mathbf{r}^{n},
{\mathbf{r}^\prime}^{n} \in A_{n}}
\delta_{f(\mathbf{r}^{n}),f({\mathbf{r}^\prime}^{n})}
P_{\mathbf{R}^{n}\mathbf{M}|z^{n}\mathbf{w}^{n}}(\mathbf{r}^{n},\mathbf{m})
\notag \\
&&~~~~~~~~~~~~~~P_{\mathbf{R}^{n}\mathbf{M}|z^{n}\mathbf{w}^{n}}({\mathbf{r}^\prime}^{n},\mathbf{m}),
\end{eqnarray}
where $\delta_{f(\mathbf{r}^{n}),f({\mathbf{r}^\prime}^{n})}$ is
 Kronecker's delta.  On the other hand, we can rewrite the second term
 of Eq.~(\ref{eq-lemma-pa-7}) as
\begin{eqnarray}
\lefteqn{ \sum_{\mathbf{m}} \frac{1}{|{\cal S}|}
P_{\mathbf{R}^{n}\mathbf{M}|z^{n}\mathbf{w}^{n}}(A_{n},\mathbf{m})^2}
\notag\\ &=& \sum_{\mathbf{m}} \sum_{\mathbf{r}^{n},
{\mathbf{r}^\prime}^{n} \in A_{n}} \frac{1}{|{\cal S}|}
P_{\mathbf{R}^{n}\mathbf{M}|z^{n}\mathbf{w}^{n}}(\mathbf{r}^{n},\mathbf{m})
\notag \\
&& ~~~~~P_{\mathbf{R}^{n}\mathbf{M}|z^{n}\mathbf{w}^{n}}({\mathbf{r}^\prime}^{n},\mathbf{m}).
\end{eqnarray}
Thus, averaging Eq.~(\ref{eq-lemma-pa-7}) over $f$, we obtain
\begin{eqnarray}
&& \sum_{\mathbf{m}} \sum_{\mathbf{r}^{n}, {\mathbf{r}^\prime}^{n} \in
 A_{n}} \mathbb{E}_f \left[
 \delta_{f(\mathbf{r}^{n}),f({\mathbf{r}^\prime}^{n})} - \frac{1}{|{\cal
 S}|} \right] \notag \\
&&  P_{\mathbf{R}^{n}\mathbf{M}|z^{n}\mathbf{w}^{n}}(\mathbf{r}^{n},\mathbf{m})
 P_{\mathbf{R}^{n}\mathbf{M}|z^{n}\mathbf{w}^{n}}({\mathbf{r}^\prime}^{n},\mathbf{m}). \label{eq-lemma-pa-8}
\end{eqnarray}
Since $f$ is chosen from a universal-hash-family, we obtain
\begin{eqnarray*}
\mathbb{E}_f \left[
	\delta_{f(\mathbf{r}^{n}),f({\mathbf{r}^\prime}^{n})} -
	\frac{1}{|{\cal S}|} \right] \le \left\{
	\begin{array}{ll} 
	1 & \mbox{for } \mathbf{r}^{n} = {\mathbf{r}^\prime}^{n} \\ 0 &
	\mbox{for } \mathbf{r}^{n} \neq {\mathbf{r}^\prime}^{n}
	\end{array} \right. 
\end{eqnarray*}
from its definition (shown in Eq.~(\ref{eq-two-universal})).  Thus,
 Eq.~(\ref{eq-lemma-pa-8}) is upper bounded by
\begin{eqnarray}
\lefteqn{ \sum_{\mathbf{m}} \sum_{\mathbf{r}^{n} \in A_{n}}
P_{\mathbf{R}^{n}\mathbf{M}|z^{n}\mathbf{w}^{n}}(\mathbf{r}^{n},\mathbf{m})
P_{\mathbf{R}^{n}\mathbf{M}|z^{n}\mathbf{w}^{n}}(\mathbf{r}^{n},\mathbf{m})}
\notag\\ &\le& \sum_{\mathbf{m}} \sum_{\mathbf{r}^{n} \in A_{n}}
P_{\mathbf{R}^{n}\mathbf{M}|z^{n}\mathbf{w}^{n}}(\mathbf{r}^{n},\mathbf{m})
\frac{1}{2^{\alpha n}} \label{eq-lemma-pa-9} \\ &\le&
\sum_{\mathbf{r}^{n},\mathbf{m}}
P_{\mathbf{R}^{n}\mathbf{M}|z^{n}\mathbf{w}^{n}}(\mathbf{r}^{n},\mathbf{m})
\frac{1}{2^{\alpha n}} \\ &=& \frac{1}{2^{\alpha n}},
\end{eqnarray}
where Eq.~(\ref{eq-lemma-pa-9}) follows from the fact that
$P_{\mathbf{R}^{n}\mathbf{M}|z^{n}\mathbf{w}^{n}}(\mathbf{r}^{n},\mathbf{m})
 \le P_{\mathbf{R}^{n}|z^{n}\mathbf{w}^{n}}(\mathbf{r}^{n}) \le
\frac{1}{2^{\alpha n}}$ for any $\mathbf{r}^{n} \in A_{n}$.  Since the
root function $\sqrt{\cdot}$ is concave function, by combining
Eqs.(\ref{eq-lemma-pa-2})--(\ref{eq-lemma-pa-9}) and averaging over $f$,
we obtain
\begin{eqnarray}
\lefteqn{\mathbb{E}_f [ \Delta_{f,z^{n},\mathbf{w}^{n}} ]}\notag \\ 
&\le&
\sqrt{|{\cal S}| |{\cal M}_{A}\times{\cal M}_{B}|
\sum_{s,\mathbf{m}}h_{n}(s,\mathbf{m})^{2}} \notag\\ 
&&+ 2
P_{\mathbf{R}^{n}|z^{n}\mathbf{w}^{n}}( \{ \mathbf{r}^{n} \in
{\cal X}_{\Delta}^{n}\times{\cal Y}_{\Delta}^{n} \mid \notag \\
&& -\frac{1}{n}\log
P_{\mathbf{R}^{n}|z^{n}\mathbf{w}^{n}}(\mathbf{r}^{n}) < \alpha \}) \notag\\ 
&\le&\sqrt{\frac{|{\cal S}| |{\cal M}_{A}\times{\cal
M}_{B}|}{2^{\alpha n}}} \notag\\ &&+ 2
P_{\mathbf{R}^{n}|z^{n}\mathbf{w}^{n}} (\{ \mathbf{r}^{n} \in
{\cal X}_{\Delta}^{n}\times{\cal Y}_{\Delta}^{n} \mid \notag \\
&& -\frac{1}{n}\log
P_{\mathbf{R}^{n}|z^{n}\mathbf{w}^{n}}(\mathbf{r}^{n}) < \alpha \}) .
\end{eqnarray}
\end{proof}
\begin{corollary}
Suppose that we set $\frac{1}{n}\log|{\cal S}| =
H(\mathbf{R}|Z\mathbf{W}) - \frac{1}{n} \log|{\cal M}_{A}\times{\cal
M}_{B}| - 2 \delta $, $\mathbb{E}_f[\Delta_{f}]$ is exponentially small
for sufficiently large $n$.
\end{corollary}

\begin{proof}
Suppose that we set $\alpha = H(\mathbf{R}|Z\mathbf{W}) - \delta$ for
$\delta >0$, the second term of Eq.~(\ref{eq-lemma-pa-0}) exponentially
tends to $0$ as $n \to \infty$ by using the Chernoff bound \cite{cover}.
On the other hand, suppose that we set $\frac{1}{n}\log|{\cal S}| =
H(\mathbf{R}|Z\mathbf{W}) - \frac{1}{n} \log|{\cal M}_{A}\times{\cal
M}_{B}| - 2 \delta $, the first term of Eq.~(\ref{eq-lemma-pa-0}) is
$e^{-\delta n}$ and tends to $0$ as $n \to \infty$.  Thus, suppose that
we set $\frac{1}{n}\log|{\cal S}| = H(\mathbf{R}|Z\mathbf{W}) -
\frac{1}{n} \log|{\cal M}_{A}\times{\cal M}_{B}| - 2 \delta $,
$\mathbb{E}_f [ \Delta_f ]$ exponentially tends to $0$ as $n \to
\infty$.
\end{proof}
If $\mathbb{E}_f [ \Delta_f ]$ is exponentially small, then the security
of the protocol in the sense of entropy is guaranteed by lemma 3.  From
this fact and corollary 5, suppose that we set $\frac{1}{n}\log|{\cal
S}| < H(X_{\Delta}Y_{\Delta}|ZW_{A}W_{B}) - \frac{1}{n} \log|{\cal
M}_{A}\times{\cal M}_{B}|$, Eq.~(\ref{eq:security}) is satisfied for
sufficiently large $n$.




\end{document}